\documentclass[11pt]{article}

\usepackage{palatino}
\usepackage{mathpazo}
\usepackage{braket}
\usepackage{amsfonts}
\usepackage{amssymb}
\usepackage{amsmath}
\usepackage{latexsym}
\usepackage{amsthm}
\usepackage[usenames]{color}
\usepackage{hyperref}

\hypersetup{pdfpagemode=UseNone}

\newtheorem{theorem}{Theorem}[section]
\newtheorem{lemma}[theorem]{Lemma}

\newtheorem{definition}[theorem]{Definition}

\newenvironment{mylist}[1]{\begin{list}{}{
	\setlength{\leftmargin}{#1}
	\setlength{\rightmargin}{0mm}
	\setlength{\labelsep}{2mm}
	\setlength{\labelwidth}{8mm}
	\setlength{\itemsep}{0mm}}}
	{\end{list}}

\newcommand{\comment}[1]{}

\newcommand{\tinyspace}{\mspace{1mu}}

\newcommand{\abs}[1]{\left\lvert\tinyspace #1 \tinyspace\right\rvert}

\newcommand{\norm}[1]{\left\lVert\tinyspace #1 \tinyspace\right\rVert}
\newcommand{\tr}{\operatorname{Tr}}
\newcommand{\ip}[2]{\left\langle #1 , #2\right\rangle}
\newcommand{\setft}[1]{\mathrm{#1}}
\newcommand{\lin}[1]{\setft{L}\left(#1\right)}
\newcommand{\density}[1]{\setft{D}\left(#1\right)}

\newcommand{\herm}[1]{\setft{Herm}\left(#1\right)}
\newcommand{\pos}[1]{\setft{Pos}\left(#1\right)}
\newcommand{\sep}[1]{\setft{Sep}\left(#1\right)}
\newcommand{\class}[1]{\textup{#1}}

\newcommand{\inner}[2]{\langle #1, #2 \rangle}

\newcommand{\calA}{\mathcal{A}}

\newcommand{\C}{\mathbb{C}}
\newcommand{\R}{\mathbb{R}}
\newcommand{\Z}{\mathbb{Z}}

\newcommand{\Span}{\mathrm{Span}}

\newcommand{\interior}{\mathrm{int}}
\newcommand{\Pos}{\mathrm{Pos}}

\newcommand{\Sep}{\mathrm{Sep}}

\newcommand{\ketbra}[2]{\ket{#1} \bra{#2}}

\newcommand{\trnorm}[1]{\norm{#1}_{\mathrm {tr}}}  
\newcommand{\snorm}[1]{\norm{#1}_{\mathrm {\infty}}}    

\newcommand{\E}{\mathbb{E}}

\def\complex{\mathbb{C}}

\def\natural{\mathbb{N}}
\def\I{\mathbb{1}}
\def\({\left(}
\def\){\right)}
\def\X{\mathcal{X}}
\def\Y{\mathcal{Y}}
\def\Z{\mathcal{Z}}

\def\yes{\text{yes}}
\def\no{\text{no}}
\def\poly{\textup{poly}}
\def\blog{\textup{log}}
\newcommand{\ve}[1]{\mathbf{#1}}

\begin{document}

\title{\bf QMA variants with polynomially many provers}

\author{%
  Sevag Gharibian\thanks{%
Department of Computer Science, University of Illinois. 
Email: {
\tt
sggharib@cs.uwaterloo.ca}.}
  \and
  Jamie Sikora\thanks{%
Laboratoire d'Informatique Algorithmique: Fondements et Applications, Universit\'e Paris Diderot.
Email: {
\tt
jwjsikor@uwaterloo.ca}.}
  \and
Sarvagya Upadhyay\thanks{%
Centre for Quantum Technologies, National University of Singapore.
Email: {
\tt
sarvagya@nus.edu.sg}.}
}

\date{\quad \\ September 3, 2012}

\maketitle

\vspace{-5mm}

\begin{abstract}
We study three variants of multi-prover quantum Merlin-Arthur proof systems. We first show that the class of problems that can be efficiently verified using polynomially many quantum proofs, each of logarithmic-size, is exactly \class{MQA} (also known as QCMA), the class of problems which can be efficiently verified via a classical proof and a quantum verifier. We then study the class $\class{BellQMA}(\poly)$, characterized by a verifier who first applies unentangled, nonadaptive measurements to each of the polynomially many proofs, followed by an arbitrary but efficient quantum verification circuit on the resulting measurement outcomes. We show that if the number of outcomes per nonadaptive measurement is a polynomially-bounded function, then the expressive power of the proof system is exactly \class{QMA}. Finally, we study a class equivalent to \class{QMA}($m$), denoted $\class{SepQMA}(m)$, where the verifier's measurement operator corresponding to outcome {\it accept} is a fully separable operator across the $m$ quantum proofs. Using cone programming duality, we give an alternate proof of a result of Harrow and Montanaro [FOCS, pp. 633--642 (2010)] that shows a perfect parallel repetition theorem for $\class{SepQMA}(m)$ for any $m$.
\end{abstract}

\section{Introduction and summary of results}\label{scn:intro}

The study of classical proof systems has yielded some of the greatest achievements in theoretical computer science, from the Cook-Levin theorem~\cite{C71,L73}, which formally ushered in the age of NP verification systems and the now ubiquitous notion of NP-hardness, to the more modern PCP theorem~\cite{AS98,ALMSS98}, which has led to significant advancements in our understanding of hardness of approximation. A natural generalization of the class \class{NP}, or more accurately its probabilistic cousin Merlin-Arthur (\class{MA}), to the quantum setting is the class quantum Merlin-Arthur (QMA)~\cite{KSV02}, where a computationally powerful but untrustworthy prover, Merlin, sends a \emph{quantum} proof to convince an efficient \emph{quantum} verifier, Arthur, that a given input string $x \in \{ 0,1 \}^n$ is a YES-instance for a specified promise problem.

More specifically, a QMA proof system for a given promise problem $A$ is characterized by the following properties (see Section~\ref{sscn:complexityclasses} for formal definitions):
\begin{itemize}
\item For every YES-instance $x$ of $A$, there exists a polynomial-size quantum proof which can convince Arthur of this fact with high probability, with the smallest such success probability over all YES-instances called the \emph{completeness} of the protocol.
\item For every NO-instance $x$ of $A$ and for \emph{any} purported quantum proof, Arthur rejects with high probability, with the maximum success probability over all NO-instances called the \emph{soundness} of the protocol.
\end{itemize}
\noindent It is easy to see that QMA proof systems are at least as powerful as \class{NP} or \class{MA}, since the ability to process and exchange quantum information does not prevent Arthur from choosing to act classically.

Much attention has been devoted to \class{QMA} over recent years. We now have a number of problems which are \emph{complete} for \class{QMA} (see e.g.~\cite{B06,L06,BS07,LCV07,SV09,JGL09,WMN10,R09}), with the quantum analogue of classical constraint satisfaction,  the physically motivated $k$-local Hamiltonian problem~\cite{KSV02,KR03,KKR06,OT05,AGIK09}, being the canonical \class{QMA}-complete problem. In analogy with \class{NP}-complete problems, it is tempting to think of \class{QMA}-complete problems as hard even for a \emph{quantum} computer to solve, though this is somewhat of a misnomer as even \class{NP}-complete problems are generally believed to be intractable for quantum computers. \class{QMA} is an extremely robust complexity class that satisfies strong error-reduction properties, and using these properties one can, e.g., give a very elegant and simple proof that $\class{MA} \subseteq \class{QMA} \subseteq \class{PP}$ (the first containment follows trivially from the definition)~\cite{MW05}. However, there still remain important open questions --- for example, despite the fact that $\class{MA}$ is contained in the polynomial hierarchy ($\class{PH}$)~\cite{AB09}, we do not even know whether $\class{BQP} \subseteq\class{PH}$. 

An approach for understanding a complexity class is to consider how introducing variations to its definition changes its properties. In this paper, we thus ask: \emph{How does allowing multiple unentangled provers affect the expressive power of \class{QMA}?} In particular, we are interested in variants of the class $\class{QMA}(\poly)$, a.k.a. quantum Merlin-Arthur proof systems with polynomially many Merlins, where the verifier receives a polynomial number of quantum proofs, which are promised to be unentangled with each other. Note that the \emph{classical} version of this class collapses trivially to \class{MA}, as the set of potential strategies of a single Merlin and the set of potential strategies of multiple Merlins coincide. This logic fails, however, in the quantum case, as a single Merlin simulating the action of multiple Merlins can try to cheat by entangling the multiple proofs. Despite much effort, very little is known (more details under \emph{Previous Work} below) about the structural properties of  $\class{QMA}(\poly)$, except for the obvious containments $\class{QMA}\subseteq\class{QMA}(\poly)\subseteq\class{NEXP}$.

\paragraph{Our results:} We show the following three results regarding variants of $\class{QMA}(\poly)$.


\paragraph{1. A complete characterization in the logarithmic-size message setting.} Let $\class{QMA}_{\log}(\poly)$ denote the restriction of the class $\class{QMA}(\poly)$ to the setting where each prover's proof is at most a \emph{logarithmic} number of quantum bits, or \emph{qubits}. We show (for $\class{MQA}$ defined below):
\begin{theorem}\label{thm:qmalog}
    $\class{QMA}_{\log}(\poly)=\class{MQA}$.
\end{theorem}

\noindent Here, \class{MQA}, also known as QCMA in the literature~\cite{AN02,JW06,A06,AK07,Beigi08,ABBS08,WY08} (the name \class{MQA} was suggested by Watrous~\cite{W09}), is defined as $\class{QMA}$ except Merlin's proof is a polynomial-size \emph{classical} string. Theorem~\ref{thm:qmalog} says that if each prover is restricted to sending short quantum proofs, then one can not only do away with multiple provers, but also of the need for \emph{quantum} proofs altogether.




\paragraph{2. Towards a non-trivial upper bound on $\class{BellQMA}(\poly)$.} Another approach to studying the question of whether $\class{QMA}=\class{QMA(\poly)}$ is to understand the properties of restricted versions of $\class{QMA}(\poly)$, and this is precisely where the class \class{BellQMA}(\poly) comes into play. $\class{BellQMA}(\poly)$ is defined~\cite{B08,ABDFS09,CD10} analogously to $\class{QMA}(\poly)$, except that before applying his quantum verification circuit to the polynomially many unentangled quantum proofs, Arthur must measure each proof using a nonadaptive and unentangled (across all proofs) measurement (we call this \emph{Stage 1} of the verification). He then feeds the resulting \emph{classical} outcomes induced by these measurements into his arbitrary efficient quantum circuit (we call this \emph{Stage 2}). This quantum circuit implements a two-outcome measurement operation corresponding to outcomes {\it accept} and {\it reject}.

The significance of $\class{BellQMA}(\poly)$ here is that if $\class{QMA} \neq \class{BellQMA}(\poly)$, then it follows that $\class{QMA}\neq\class{QMA}(\poly)$, since $\class{QMA}\subseteq\class{BellQMA}(\poly)\subseteq\class{QMA}(\poly)$. To this end, Brand\~{a}o has shown the negative result that for \emph{constant} $m$, $\class{QMA}=\class{BellQMA}(m)$~\cite{B08}. Where the class $\class{BellQMA}(\poly)$ lies, however, remains open. For example, the techniques used to show $\class{QMA}(2)=\class{QMA}(\poly)$ \cite{HM10} do not straightforwardly extend to show the analogous result $\class{BellQMA}(2)=\class{BellQMA}(\poly)$ as they require entangled measurements (i.e. SWAP test measurements) across multiple proofs, which violate the definition of $\class{BellQMA}$.

To make progress on $\class{BellQMA}(\poly)$, we introduce the class $\class{BellQMA}[r,m]$, which is defined as $\class{BellQMA}(m)$ with $m$ provers and the additional restriction that in Stage 1 above, the number of outcomes per proof in Arthur's nonadaptive measurements is upper bounded by $r$. Our contribution is the following:

\begin{theorem}\label{thm:bellqma}
For any polynomially bounded functions $r,m:\natural \rightarrow \natural$, we have the containment $\class{BellQMA}[r,m] \subseteq \class{QMA}$ (where the containment holds with equality when $r\geq 2$).
\end{theorem}

\noindent In other words, $\class{BellQMA}(\poly)$ cannot be used to show $\class{QMA}\neq\class{QMA}(\poly)$ if the verifier in the $\class{BellQMA}(\poly)$ protocol is restricted to have a polynomially bounded number of measurement outcomes per proof in Stage 1. We remark that, in general, the number of such measurement outcomes can be exponential in the input length --- the restriction that $r$ be a polynomially bounded function is crucial for the proof of Theorem~\ref{thm:bellqma}. For this reason, our result complements, rather than subsumes Brand\~{a}o's result~\cite{B08}. In other words, in our notation, Brand\~{a}o has shown that $\class{BellQMA}[\exp,\rm{const}] = \class{QMA}$, and we show $\class{BellQMA}[\poly,\poly] = \class{QMA}$.

Note that we allow the second stage of the verification procedure above to be {\it quantum}, as per the definition suggested by Chen and Drucker~\cite{CD10}, as opposed to {\it classical}, as studied by Brand\~{a}o~\cite{B08}. The conclusion of Theorem~\ref{thm:bellqma} holds even if the second stage of verification is completely classical.

Finally, it is worth noting that by combining Theorems~\ref{thm:qmalog} and~\ref{thm:bellqma}, we conclude that in the setting of $\class{BellQMA}(\poly)$, if $\class{MQA} \neq \class{QMA}$, then having the Merlins send logarithmic-size proofs without any restriction on the number of local measurement outcomes of Arthur in Stage 1 has less expressive power than sending polynomial-size proofs but restricting the number of outcomes, even though the number of measurement outcomes in Stage 1 per Merlin in both cases is the same, i.e. polynomial in the input length.

\paragraph{3. Perfect parallel repetition for $\class{SepQMA}(m)$.} A key question in designing proof systems is how to improve the completeness and soundness parameters of a verification protocol without increasing the required number of rounds of communication. A natural approach for doing so is to repeat the protocol multiple times in parallel. With \class{QMA}, however, this raises the concern that Merlin might try to cheat by entangling his proofs across these parallel runs. If, though, \emph{perfect parallel repetition} holds, it means that for any input string $x$, if the verification procedure $V$ accepts with probability $p(|x|)$, then if we run $V$ $k$ times in parallel, the probability of accepting in all $k$ runs of $V$ is precisely $p(|x|)^k$. Note that we do not put any restriction on the quantum proof, which can be entangled across the $k$ executions of the protocol. In other words, if perfect parallel repetition holds, there is no incentive for Merlin to cheat --- an honest proof which is a product state across all $k$ runs achieves the maximum success probability.

Our final contribution is an alternate proof of a perfect parallel repetition theorem for a class which is equivalent~\cite{HM10} to \class{QMA}($m$), namely $\class{SepQMA}(m)$. The theorem was first proved in Harrow and Montanaro~\cite{HM10} in connection with an error reduction technique for $\class{QMA}(\poly)$. However, our proof is significantly different from theirs and uses the cone programming characterization of $\class{QMA}(\poly)$. Here $\class{SepQMA}(m)$ is defined as \class{QMA}($m$) with the restriction that Arthur's measurement operator corresponding to acceptance is a \emph{separable} operator across the $m$ unentangled proofs. (Note that this does not imply that Arthur's measurement operator corresponding to rejection is also separable.) We show:

\begin{theorem}
[see~\cite{HM10} for alternate proof]\label{thm:parrep} The class $\class{SepQMA}(m)$ admits perfect parallel repetition.
\end{theorem}

\noindent Our alternate proof of Theorem~\ref{thm:parrep} is significant in that, to the best of our knowledge, it is the first use of duality theory for a cone program \emph{other} than a semidefinite program to establish a parallel repetition result (note that cone programming generalizes semidefinite programming). We remark that semidefinite programs have been previously used to show perfect or strong parallel repetition theorems for various other models of (single or two-prover) quantum interactive proof systems~\cite{CSUU08, KRT10, G09}, and that the alternate proof of Theorem~\ref{thm:parrep} of Harrow and Montanaro is not based on cone programming. Perfect parallel repetition for $\class{SepQMA}(m)$ in itself is interesting, as it has been used to show that error reduction is possible for $\class{QMA}(m)$ proof systems~\cite{HM10}.

\paragraph{Proof ideas and tools:} The proof of our first result, Theorem~\ref{thm:qmalog}, is simple, and is an application of the facts that (1) quantum states of a logarithmic number of qubits can be described to within inverse exponential precision using a polynomial number of classical bits, and conversely that (2) given such a classical description, a logarithmic-size quantum state can be efficiently prepared by a quantum circuit. Hence, roughly speaking, one can replace a polynomial number of logarithmic-size quantum proofs with a single polynomial size classical proof, thereby avoiding the danger of a cheating Merlin using entanglement. Although the proof is simple, one cannot hope for a better characterization using other techniques because the reverse containment $\class{MQA}\subseteq \class{QMA}_{\log}(\poly)$ holds using similar ideas.

More technically challenging is our second result, Theorem~\ref{thm:bellqma}. To show the containment $\class{BellQMA}[\poly,\poly]\subseteq \class{QMA}$ (note that the reverse containment $\class{QMA}\subseteq\class{BellQMA}[\poly,\poly]$ is trivial since $\class{QMA} = \class{BellQMA}[2,1]$), we demonstrate a \class{QMA} protocol which simulates an arbitrary $\class{BellQMA}[\poly,\poly]$ protocol using the following observation: {Although consolidating $m$ quantum proofs into a single quantum proof raises the possibility of cheating using entanglement, if Arthur is {also} sent an appropriate classical ``consistency-check'' string, then a dishonest Merlin can be caught with non-negligible probability.}

Specifically, in our \class{QMA} protocol, we ask a single Merlin to send the $m$ quantum proofs of the original \class{BellQMA} protocol (denoted by a single state $\ket{\psi}$), accompanied by a  ``consistency-check'' string $\ve{p}$ which is a classical description of the probability distributions obtained as the output of Stage 1. One can think of this as having the \class{QMA} verifier \emph{delegate} Stage 1 of the \class{BellQMA} verification to Merlin. Arthur then performs a consistency check between $\ket{\psi}$ and $\ve{p}$ based on the premise that if Merlin is honest, then $\ve{p}$ should arise from running Stage 1 of the original verification on $\ket{\psi}$. If this check passes, then Arthur runs Stage 2 of the \class{BellQMA} verification on $\ve{p}$. If Merlin tries to cheat, however, we show that the check detects this with non-negligible probability. Note that the accuracy of the consistency check crucially uses the fact that there are at most polynomially many outcomes to check for each local measurement of Stage 1.

Finally, our last result, Theorem~\ref{thm:parrep}, is shown using duality theory for a class of cone programs that captures the success probability of a $\class{QMA}(\poly)$ protocol. In particular, we phrase the maximum acceptance probability of a (possibly cheating) prover for the two-fold repetition of a $\class{SepQMA}(m)$ verification protocol as a cone program. We then demonstrate a feasible solution for its dual yielding an upper bound on the maximum acceptance probability. The objective value of this dual solution is precisely the product of the optimum values of the two instances of the $\class{SepQMA}(m)$ verification protocols. We conclude that one of the optimal strategies of the provers is to be faithful in the following sense: Each prover elects not to entangle his/her two quantum proofs for the two instances of the $\class{SepQMA}(m)$ protocol and instead sends a tensor product of optimal proofs for both the instances.

\paragraph{Previous work.} The expressive power of multiple Merlins was first studied by Kobayashi, Matsumoto and Yamakami~\cite{KMY03}, who showed that $\class{QMA}(2) = \class{QMA}(\poly)$ if and only if the class of \class{QMA}(2) protocols with completeness $c$ and soundness $s$ (with at least inverse polynomial gap) is exactly equal to $\class{QMA}(2)$ protocols with completeness $2/3$ and soundness $1/3$. A substantial amount of research has since been devoted to understanding the properties of multi-prover quantum Merlin-Arthur proof systems. Recently, Harrow and Montanaro~\cite{HM10} demonstrated a {\it product state test}, wherein given two copies of a {\it pure} quantum state on multiple systems, the test distinguishes between the cases when the quantum state is a {\it fully} product state across all the systems or {\it far} from any such state. Using this test, they answered a few important questions regarding $\class{QMA}(\poly)$. In particular, they showed that
\[
\class{QMA}(2) = \class{QMA}(\poly)
\]
and that error reduction is possible for such proof systems. Prior to their result, the answers to both the questions were known to be affirmative assuming a {\it weak} version of the Additivity Conjecture~\cite{ABDFS09}. One of the crucial properties of the product state test is that it can be converted into a $\class{QMA}(2)$ protocol, where Arthur's measurement operator corresponding to outcome {\it accept} is a separable operator across the two proofs. Harrow and Montanaro established a perfect parallel repetition theorem for such proof systems, a crucial step in obtaining exponentially small error probabilities.

Blier and Tapp initiated the study of \emph{logarithmic}-size unentangled quantum proofs~\cite{BT10}. 
They showed that two unentangled quantum proofs suffice to show that a 3-coloring of an input graph exists, implying that $\class{NP}$ has \emph{succinct} unentangled quantum proofs. A drawback of their protocol is that although it has {\it perfect} completeness, its soundness is only inverse polynomially bounded away from $1$. Shortly after, Aaronson, Beigi, Drucker, Fefferman and Shor~\cite{ABDFS09} showed that satisfiability of any 3-SAT formula of size $n$ can be proven by $\widetilde{O}(\sqrt{n})$ unentangled quantum proofs of $O(\log n)$ qubits with perfect completeness and constant soundness (see also~\cite{CD10}). In a subsequent paper~\cite{Beigi08}, Beigi improved directly on Blier and Tapp's result~\cite{BT10} by showing that by sacrificing perfect completeness, one can show that \class{NP} has two logarithmic-size quantum proofs with a better gap between completeness and soundness probabilities than in~\cite{BT10}. Very recently, Chiesa and Forbes showed a better completeness and soundness gap of $\Omega\left(\frac{1}{n^2}\right)$ for the Blier and Tapp protocol~\cite{CF11}. Also, Le Gall, Nakagawa, and Nishimura showed that 3-SAT has a $\class{QMA}_{\log}(2)$ proof system with completeness~$1$ and soundness $1 - \Omega\left(\frac{1}{n\poly\blog(n)}\right)$~\cite{LNN12}.

Finally, one of the open questions raised in Ref.~\cite{ABDFS09} concerns  the power of Arthur's verification procedure. In particular, the paper introduces two different classes of verification procedures, \class{BellQMA} and \class{LOCCQMA} verification. Roughly speaking, \class{LOCCQMA} verification corresponds to Arthur applying a measurement operation that can be implemented by Local Operations and Classical Communication (LOCC) (with respect to the partition induced by the multiple proofs). The authors raised the question of whether $\class{BellQMA}(\poly) = \class{QMA}$ or not. Brand\~{a}o~\cite{B08} showed that $\class{BellQMA}(m)$ is equal to \class{QMA} for constant $m$. In a recent development, Brand\~{a}o, Christandl and Yard~\cite{BCY11} showed that $\class{LOCCQMA}(m)$ is equal to \class{QMA} for constant $m$.

\paragraph{Organization of this paper.} We begin in Section~\ref{scn:preliminaries} with background and notation, defining relevant complexity classes in Section~\ref{sscn:complexityclasses}, and reviewing cone programming in Section~\ref{sscn:conic}. Theorems~\ref{thm:qmalog},~\ref{thm:bellqma}, and~\ref{thm:parrep} are proved in Sections~\ref{scn:qmalog},~\ref{scn:bellqma}, and~\ref{scn:parrep}, respectively. We conclude with open problems in Section~\ref{scn:conclusions}.

\section{Preliminaries and Notation}\label{scn:preliminaries}

We begin by setting our notation, and subsequently review the background material required for this paper. First, the notation $[m]$ indicates the set $\{1, \ldots, m\}$, and $\abs{x}$ the length of a string $x \in \{ 0,1 \}^\ast$. We let uppercase script letters $\X,\Y,\Z$ denote complex Euclidean spaces. We denote the sets of linear, Hermitian, positive semidefinite, and density operators acting on vector space $\X$ by $\lin{\X}$, $\herm{\X}$, $\pos{\X}$, and $\density{\X}$, respectively. We denote the standard Hilbert-Schmidt inner product of operators $A$ and $B$ as $\ip{A}{B}:=\tr(A^\ast B)$, where $A^*$ denotes the adjoint of $A$. The spectral and trace norms of an operator $A$ are given by $\snorm{A} := \max\{\norm{Au} : \norm{u} = 1\}$ and $\trnorm{A} := \tr\big(\sqrt{A^{\ast}A}\big)$, respectively, where $\norm{u}$ denotes the Euclidean norm of a vector $u$. These can be thought of as the largest singular value and the sum of singular values of $A$, respectively. A useful lemma in this paper regarding the trace norm is the following:

\begin{lemma}[\cite{W02}]\label{lem:tracenorm}
Let $\{\rho_1 \dots, \rho_k\} \subset \density{\X}$ and $\{\sigma_1, \dots, \sigma_k\} \subset \density{\X}$. Then
\[
\bigg\Vert \bigotimes_{i=1}^k \rho_i - \bigotimes_{i=1}^k \sigma_i \bigg\Vert_{\textup{tr}} \le \sum_{i=1}^k \trnorm{\rho_i - \sigma_i}.
\]
\end{lemma}

Next, we say a (possibly unnormalized) operator $A \in \pos{\X_1 \otimes \dots \otimes \X_m}$ is {\it fully separable} if it can be written as
\[
A = \sum_{i=1}^k P_{1}(i) \otimes \dots \otimes P_{m}(i)
\]
where $P_j(i) \in \pos{\X_j}$, for every $j \in [m]$ and $i \in [k]$. The set of fully separable operators is denoted $\sep{\X_1,\X_2, \dots ,\X_m}$. This notation is helpful in the context of cone programming. In the setting of quantum information, one typically also has $\tr(A)=1$. The set of fully separable density operators is convex, compact, and has non-empty interior since it contains a ball around the normalized identity operator~\cite{GB02,GB03,GB05}. 

We use the fact that any pure quantum state $\ket{\psi} \in \complex^N$ can be described approximately classically using $N\cdot f(N)$ bits, for some function $f:\mathbb{N} \rightarrow \mathbb{N}$. The resulting approximate description $\ket{\psi'}$ satisfies $\norm{\ket{\psi}-\ket{\psi'}} \le N2^{-(f(N)+1)}$. We also speak in terms of {\it quantum registers} rather than quantum states in the next two sections. To make the association precise, an $n$-qubit quantum register $\mathsf{X}$ is associated with a vector space $\X = \complex^{2^n}$ and contains any element of $\density{\X}$.

Finally, moving to quantum operations, the notion of measurement used in this paper is that of a Positive Operator Valued Measure (POVM), given by a finite set of positive semidefinite operators $\{\Pi_1, \dots, \Pi_r\} \subset \pos{\X}$ obeying
\[
\sum_{i=1}^r \Pi_i = \I_{\X}.
\]
Regarding unitary operators, we use the fact that any unitary operator acting on $k$ qubits can be approximated within high precision by a finite set of one-qubit, two-qubit, and/or three-qubit unitary operators. Such a finite set is often referred to as an {\it approximately} universal set of quantum gates, and one such set is comprised of the Toffoli, Hadamard, and phase-shift gates. The Solovay-Kitaev theorem implies that the action of an arbitrary unitary operator $U$ on $k$ qubits can be simulated by a composition $\widetilde{U}$ of $O(4^k \, \poly\(\log(1/\epsilon)\))$ many universal gates, such that $\snorm{U - \widetilde{U}} \le \epsilon$~\cite{NC00}.

\subsection{Relevant quantum complexity classes}\label{sscn:complexityclasses}

A promise problem $A = (A_{\yes}, A_{\no})$ is a partition of the set $\{0,1\}^*$ into three disjoint subsets: the set $A_{\yes}$ denotes the set of YES-instances of the problem, the set $A_{\no}$ denotes the set of NO-instances of the problem, and the set $\{0,1\}^*\backslash (A_{\yes} \cup A_{\no})$ is the set of disallowed strings (we are \emph{promised} the input does not fall into this last set). We now define $\class{QMA}(m)$, or $\class{QMA}$ with $m$ unentangled provers.

\begin{definition}[$\class{QMA}(m)$]\label{def:qma-poly}
Let $p:\mathbb{N} \rightarrow \mathbb{N}$ be a polynomially bounded function, and $m:\mathbb{N} \rightarrow \mathbb{N}$ a function. A promise problem $A = (A_{\yes}, A_{\no})$ is in the class $\class{QMA}(m)$ if there exists a polynomial-time generated family of verification circuits $Q=\set{Q_n\mid n\in\mathbb{N}}$ with the following properties:
\begin{mylist}{\parindent}
\item [1.] Each $Q_n$ acts on $n+p(n)$ input qubits, and outputs one qubit.
\item [2.]
 (Completeness) For every $x \in A_{yes}$, there exist $m(|x|)$ quantum proofs 
 \[ \ket{\psi_1}, \ket{\psi_2}, \dots , \ket{\psi_{m(|x|)}} \in \complex^{2^{p(|x|)}} \] 
 such that
 \[
 \operatorname{Pr}[Q_{\abs{x}} \text{ accepts } (x,\ket{\psi_1}\otimes\ldots\otimes\ket{\psi_{m(|x|)}})] \geq 2/3.
 \]
\item [3.]
 (Soundness) For any $x \in A_{no}$ and any $m(|x|)$ quantum proofs 
 \[ \ket{\psi_1}, \ket{\psi_2}, \dots , \ket{\psi_{m(|x|)}} \in \complex^{2^{p(|x|)}}, \]
 we have 
 \[
 \operatorname{Pr}[Q_{\abs{x}} \text{ accepts } (x,\ket{\psi_1}\otimes\ldots\otimes\ket{\psi_{m(|x|)}})] \leq 1/3.
 \]
\end{mylist}

\noindent Furthermore, the class $\class{QMA}(\poly)$ is defined as $\class{QMA}(\poly) = \bigcup_{m \in \poly} \class{QMA}(m)$.
\end{definition}

\noindent We remark that the constants $2/3$ and $1/3$ can be replaced by any $a,b \ge 0$, respectively, such that $a-b \geq 1/\poly(n)$. This does not change the expressive power of the proof system.

All complexity classes considered in this paper are variants of $\class{QMA}(m)$ and satisfy the properties mentioned above in Definition~\ref{def:qma-poly}. We define the following variants, which are relevant to this paper.

\begin{mylist}{\parindent}
	\item [1.] \textbf{[\class{QMA} and \class{MQA}]} The class \class{QMA} is simply \class{QMA}(1). If we replace the quantum proofs in the definition of \class{QMA} with a polynomial-size classical proof string, the corresponding class is denoted \class{MQA}.

	\item [2.] \textbf{[\class{SepQMA}(\poly)]} The class $\class{SepQMA}(\poly)$ is a subclass of $\class{QMA}(\poly)$, wherein Arthur's measurement operator corresponding to outcome {\it accept} is a fully separable operator across the proofs.
	
	\item [3.] \textbf{[$\class{QMA}_{\log}(\poly)$]} The class $\class{QMA}_{\log}(\poly)$ is a subclass of $\class{QMA}(\poly)$, wherein each Merlin's message to Arthur is $O(\log (|x|))$ qubits in length.

\end{mylist}

\noindent For clarity, we give a formal definition of the variant of $\class{BellQMA}$ we introduce, $\class{BellQMA}[r,m]$.

\begin{definition}[$\class{BellQMA}{[}r,m{]}$]
Let $r,m:\natural \rightarrow \natural$ be two functions. We say that a  promise problem $A = (A_{\yes}, A_{\no})$ is in $\class{BellQMA}[r,m]$ if there exists a $\class{QMA}(m)$ verification protocol in which Arthur is restricted to act as follows.

\begin{mylist}{\parindent}
\item[1.]
Arthur performs a polynomial-time quantum computation on the input $x$ and generates a description of quantum circuits $V_1(x), \dots, V_m(x)$, one for each of the $m$ provers.

\item[2.]
(Stage 1) Arthur simultaneously measures all $m$ quantum proofs by applying $V_i(x)$ to the $i$-th quantum proof, where the action of $V_i(x)$ can be described by a unitary operator followed by measurement in the standard basis. The label of the $i$-th measurement outcome is stored as a classical string $y_i$ also identified as an element of $[r(|x|)]$.

\item[3.]
(Stage 2) Arthur runs an efficient quantum verification circuit on input $x$ and measurement outcomes $(y_1,\ldots,y_m)$ to decide whether to accept or reject.
\end{mylist}
\end{definition}

\noindent Note that the key distinction between $\class{BellQMA}[r, m]$ and $\class{BellQMA}(\poly)$ is that the former has the number of measurement outcomes in Stage 1 of the protocol bounded by $r(\abs{x})$, whereas the latter may allow exponentially many possible outcomes. Throughout this paper, we use the notation $\class{BellQMA}[\poly, \poly]$ to denote
        \[
        \class{BellQMA}[\poly, \poly] := \bigcup_{r \in \poly} \bigcup_{m \in \poly}\class{BellQMA}[r,m].
        \]
We remark that, as in~\cite{CD10}, our $\class{BellQMA}$ protocols are allowed to use a \emph{quantum} verification circuit in Stage 2, whereas originally in references~\cite{B08,ABDFS09} only classical processing of measurement outcomes $\set{y_i}$ was allowed in order to emulate the notion of a \emph{Bell experiment} performed by Arthur. We again remark that Theorem~\ref{thm:bellqma} holds even if Arthur is restricted to do classical processing on the measurement outcomes.

\subsection{Cone programming} \label{sscn:conic}

We now briefly review basic notions in conic optimization (or cone programming), which is a generalization of semidefinite optimization. We say that a set $K$ in an underlying Euclidean space is a cone if $x \in K$ implies that $\lambda x \in K$ for all $\lambda > 0$. A cone $K$ is convex if $x,y \in K$ implies that $x+y \in K$. Cone programs are concerned with optimizing a linear function over the intersection of a convex cone and an affine space. It generalizes several well-studied models of optimization including semidefinite programming ($K = \Pos (\X)$) and linear programming ($K = \R_+^n$). In this paper, we are primarily concerned with the cone of fully separable operators $\sep{\X_1,\X_2, \dots ,\X_m}$ which recall is a closed, convex cone with non-empty interior.

Associated with a cone $K$ is its dual cone $K^{\ast}$ defined as
\[
K^{\ast} = \left\{S: \ip{X}{S} \ge 0 \text{ for all } X \in K \right\}.
\]
A cone program associates the following 4-tuple $(C, b, \calA, K)$ to an optimization problem described as:
\begin{align*}
	\text{supremum:}\quad & \ip{C}{X}\\
  \text{subject to:}\quad & \calA(X) = b,\\
  & X \in K.
\end{align*}
Here $\calA: \Span(K) \to \R^{m}$ is a linear transformation. Note that the inner product is defined as in the Euclidean space. For instance, if the cone under consideration is the set of positive semidefinite or separable operators, then the inner product is the standard Hilbert-Schmidt inner product over the space of Hermitian operators. We say that the cone program is \emph{feasible} if $\{X: \calA(X) = b \} \cap K$ is non-empty and \emph{strictly feasible} if $\{X: \calA(X) = b \} \cap \interior(K)$ is non-empty, where $\interior(\cdot)$ denotes the interior of a set.

Cone programs come in primal-dual pairs:
\begin{center}
  \begin{minipage}{2in}
    \centerline{\underline{Primal problem (P)}}\vspace{-7mm}
    \begin{align*}
			\text{supremum:}\quad & \ip{C}{X}\\
  		\text{subject to:}\quad & \calA(X) = b,\\
  		& X \in K.
		\end{align*}
  \end{minipage}
  \hspace*{25mm}
  \begin{minipage}{2in}
    \centerline{\underline{Dual problem (D)}}\vspace{-7mm}
		\begin{align*}
			\text{infimum:}\quad & \ip{b}{y}\\
  		\text{subject to:}\quad & \calA^{\ast}(y) = C + S,\\
  		& S \in K^{\ast}.
  	\end{align*}	
  \end{minipage}
\end{center}
Here $\calA^*$ is the adjoint of $\calA$. A convex cone $K$ is closed if and only if $K = K^{**}$. In other words, the dual of the cone $K^\ast$ is the original cone $K$. Thus, if $K$ is not closed we need to ``order'' the primal-dual pairs since $K \neq K^{**}$ implying the dual of the dual problem is not equal to the primal problem. Since the convex cone of fully separable operators is closed, ordering the primal-dual pairs is not an issue in our case.

Similar to linear programming and semidefinite programming, cone programming has a rich duality theory.

\begin{lemma}[Weak Duality]
If $X$ is primal feasible and $(y,S)$ is dual feasible then
\[
\ip{b}{y} - \ip{C}{X} = \ip{X}{S} \geq 0.
\]
\end{lemma}

This result can be used to show upper bounds on the value of the primal problem or lower bounds on the value of the dual problem. There is also a notion of \emph{strong duality}. We say that \emph{strong duality holds for a problem} (P) if the optimal value of (P) equals the optimal value of (D) and (D) attains an optimal solution. Below we give a condition that guarantees strong duality for (P).

\begin{theorem}[Strong Duality, Version 1]\label{lem:stduality}
If \textup{(P)} is strictly feasible and the optimal value is bounded from above, then strong duality holds for \textup{(P)}, i.e., \textup{(D)} attains an optimal solution and the optimal values for \textup{(P)} and \textup{(D)} coincide.
\end{theorem}

In this paper, we are concerned with closed, convex cones with non-empty interior. Since the dual of the dual problem is the primal problem when $K$ is closed, we can use the following stronger version of strong duality.

\begin{theorem}[Strong Duality, Version 2]\label{lem:stduality2}
Suppose $K$ is a closed, convex cone. If \textup{(P)} and \textup{(D)} are both strictly feasible then strong duality holds for both problems, i.e., both problems attain an optimal solution and the optimal values coincide.
\end{theorem}

We refer the reader to the work of Tun\c cel and Wolkowicz~\cite{TW08}
and the references therein for more details on cone programming duality.

\section{Equivalence of \class{MQA} and $\class{QMA}_{\blog}(\poly)$}\label{scn:qmalog}

We now prove Theorem~\ref{thm:qmalog} which states that $\class{MQA}=\class{QMA}_{\log}(\poly)$. We first show the direction $\class{MQA} \subseteq \class{QMA}_{\log}(\poly)$. Let $A = (A_{\yes},A_{\no})$ be a promise problem in \class{MQA} and let $x \in \{ 0,1 \}^n$ be the input string. Suppose the \class{MQA} prover sends an $m$-bit classical proof to the verifier, for polynomially bounded $m$. Then the following simple $\class{QMA}_{\log}(m)$ protocol achieves the desired containment:

\begin{center}
\underline{$\class{QMA}_{\log}(m)$ Protocol}
\end{center}

\begin{mylist}{\parindent}
\item[1.]
\textbf{Embed classical bits into qubits.} Each (unentangled) prover $i\in[m]$ sends a single qubit $\ket{\psi_i}\in \complex^2$ to Arthur. If the $i$-th prover is honest, his/her qubit is the computational basis state corresponding to the $i$-th bit of the classical \class{MQA} proof.

\item[2.] \textbf{Make things classical again.} Arthur measures all proofs in the computational basis, obtaining a classical string $y \in \{0,1\}^{m}$.

\item[3.]
\textbf{Run MQA verification.} Arthur runs the \class{MQA} verification circuit on $x$ and $y$ and accepts if and only if acceptance occurs in the \class{MQA} verification.
\end{mylist}

\noindent The completeness property follows straightforwardly. The soundness property is also easy to observe. Note that Arthur runs the \class{MQA} verification on a classical string $y$ and hence he accepts the string with probability at most $1/3$.

To show the reverse containment, let $A = (A_{\yes},A_{\no})$ be a promise problem in $\class{QMA}_{\log}(\poly)$ and let $x \in \{ 0,1 \}^n$ be the input string. Suppose we have a $\class{QMA}_{\log}(m)$ protocol for polynomially bounded $m$, where prover $i$ sends a $\lceil c\log n\rceil$-qubit state $\ket{\psi_i}$ for some constant $c>0$. Let
\[
r(n) = 2^{\lceil c\log n \rceil} = O(n^c).
\]
The \class{MQA} protocol proceeds as follows:

\begin{center}
\underline{\class{MQA} Protocol}
\end{center}

\begin{mylist}{\parindent}
\item[1.]
\textbf{Describe proofs classically.} The prover sends $m$ classical registers represented by the tuple $(\mathsf{C}_1, \mathsf{C}_2, \dots, \mathsf{C}_m)$, each of length $2n\cdot r(n)$ to Arthur. If the prover is honest, register $\mathsf{C}_i$ contains a classical description of the $i$-th quantum proof of the $\class{QMA}_{\log}(m)$ protocol.

\item[2.]
\textbf{State preparation.} Using the contents of register $\mathsf{C}_i$, for every choice of $i \in [m]$, Arthur prepares the state $\ket{\psi_i}$ by first determining a unitary $U_i$ such that $U_i\ket{0\ldots 0}=\ket{\psi_i}$, and then implementing $U_i$ with high precision using a finite set of approximately universal gates, obtaining states $\ket{\psi_i^\prime}$.

\item[3.]
\textbf{Run $\class{QMA}_{\log}(m)$ verification.} Arthur runs the $\class{QMA}_{\log}(m)$ verification circuit on the state $\ket{\psi_1^\prime}\otimes\cdots\otimes\ket{\psi_m^\prime}$ and accepts if and only if acceptance occurs in the $\class{QMA}_{\log}(m)$ verification.
\end{mylist}

\noindent Observe that each classical register $\mathsf{C}_i$ is of size polynomial in $n$, implying the overall proof length is of polynomial size. In Step~1, the prover uses $n$ bits to represent the real and imaginary parts of each of the polynomially many entities ($r(n)$ entries) required to describe each $\ket{\psi}$. Let the unit vector described by register $\mathsf{C}_i$ be denoted $\ket{\psi_i}$. In Step~2, $U_i$ is easily found as the unitary that maps $\ket{0\ldots0}$ to $\ket{\psi_i}$ as the inverse of the unitary that maps $\ket{\psi_i}$ to $\ket{0\ldots0}$. Such a unitary can be easily decomposed into a product of polynomially many $2 \times 2$ rotations on an $r(n)$-dimensional real space and a diagonal unitary as follows. The first step is to convert the vector $\ket{\psi_i}$ into a real vector by applying an appropriate diagonal unitary operator. The second step is to convert the resulting real unit vector into $\ket{0\ldots0}$ by shifting the amplitudes of any standard basis other than $\ket{0\ldots0}$ to $\ket{0\ldots0}$. Each of these unitary operators can be implemented by a finite set of approximately universal gates (see Bernstein and Vazirani~\cite{BV97} for details). This step also incurs some error, which can be made exponentially small.

Since Steps 1 and 2 can be performed to within inverse exponential error, we thus can ensure $\norm{\ket{\psi_i} - \ket{\psi_i^\prime}} \le \epsilon$ for all $i \in [m]$ and for inverse exponential $\epsilon>0$. By Lemma~\ref{lem:tracenorm}, it follows that the overall precision error is at most $m\epsilon$ for polynomial $m$, and thus the completeness and soundness of the protocol are bounded from below and above by (respectively)
\[
\frac{2}{3} - m\epsilon \qquad \text{and} \qquad \frac{1}{3} + m\epsilon.
\]

Alternatively, the containment $\class{QMA}_{\log}(\poly) \subseteq \class{MQA}$ can be shown using a slightly different protocol\footnote{This protocol was mentioned to us by Richard Cleve.}, where Merlin sends classical descriptions of the quantum circuits that generate the quantum proofs from $\ket{0 \ldots 0}$ instead of classical descriptions of the proofs.

\section{Equivalence of $\class{BellQMA}[\poly, \poly]$ and $\class{QMA}$}\label{scn:bellqma}

We now show Theorem~\ref{thm:bellqma}, i.e., that $\class{BellQMA}[r, m]=\class{QMA}$ for polynomially-bounded functions $r$ and $m$. For notational convenience, let $\Pi_j(i)$ denote Arthur's $i$-th POVM element in Stage 1 of the \class{BellQMA} verification protocol for the $j$-th prover (i.e. $\sum_{i=1}^r \Pi_j(i) =\I$), where we assume without loss of generality that the number of possible outcomes is exactly $r$ for each prover, and where $j\in[m]$ for $m$ the number of provers.

We proceed as follows. Let $A = (A_{\yes}, A_{\no})$ be a promise problem, and $x$ be an input string of length $n := |x|$. As mentioned in Section~\ref{scn:intro}, the containment $\class{QMA} \subseteq \class{BellQMA}[\poly, \poly]$ follows straightforwardly since $\class{QMA} \subseteq \class{BellQMA}[2,1]$. For the reverse containment, suppose we have a $\class{BellQMA}[r, m]$ protocol for polynomially bounded functions $r,m:\natural \rightarrow \natural$ with completeness $2/3$ and soundness $1/3$. We show that this protocol can be simulated by a \class{QMA} protocol where Merlin sends the following proof to Arthur.

Merlin's proof consists of two registers $(\mathsf{X}, \mathsf{Y})$, which should be thought of as the \emph{classical} and \emph{quantum} registers, respectively. Suppose optimal proofs for the $\class{BellQMA}[r, m]$ protocol for input $x$ are given by $\rho_j$ for $j\in [m]$. Then, in the quantum register $\mathsf{Y}$, an honest Merlin should send many copies of the state $\rho_j$. Specifically, $\mathsf{Y}$ is partitioned into $m$ registers $\mathsf{Y}_j$, one for each original prover, and each $\mathsf{Y}_j$ should contain $k$ copies of $\rho_j$, for $k$ a carefully chosen polynomial. In other words, $\mathsf{Y}$ should contain the state $[\rho_1^{\otimes k}]_{\mathsf{Y}_1}\otimes\cdots\otimes[\rho_m^{\otimes k}]_{\mathsf{Y}_m}$. We further view each $\mathsf{Y}_j$ as a block of registers $(\mathsf{Y}_j^1, \dots, \mathsf{Y}_j^k)$ where $\mathsf{Y}_j^l$ should contain the $l$-th copy of $\rho_j$.

In the classical register $\mathsf{X}$, an honest Merlin prepares a quantum state in the computational basis, which intuitively corresponds to a bit string describing the $m$ classical probability distributions Arthur induces upon applying the measurement operation corresponding to Stage~1 of the \class{BellQMA} verification to each of the optimal proofs $\rho_j$, respectively. More formally, we partition $\mathsf{X}$ into $mr$ registers $\mathsf{X}_j^i$ corresponding to each of the $j\in[m]$ provers and $i\in[r]$ POVM outcomes per prover. The content of $\mathsf{X}_j^i$ should be $p_j(i) := \ip{\Pi_j(i)}{\rho_j}$, truncated to $\alpha$ bits of precision ($\alpha$ polynomially bounded), such that $\sum_{i=1}^r p_j(i)=1$. For example, if the $j$-th prover's proof was the single qubit state $\rho_j=\ketbra{0}{0}$, with $\Pi_j(1)=\ketbra{0}{0}$ and $\Pi_j(2)=\ketbra{1}{1}$, then $\mathsf{X}_j = (1, 0)$. We remark that $\mathsf{X}$ plays the role of the classical ``consistency check'' string described in Section~\ref{scn:intro}.

Of course, Merlin may elect to be dishonest and choose not to send a proof of the above form to Arthur by, e.g., sending a quantum state which is entangled across the registers $(\mathsf{X}, \mathsf{Y})$. To catch this, our \class{QMA} protocol is defined as follows:

\begin{center}
\underline{\class{QMA} Protocol}
\end{center}

\begin{mylist}\parindent
\item[1.] Merlin sends Arthur a quantum state in registers $(\mathsf{X},\mathsf{Y})$, for $\mathsf{X}$ and $\mathsf{Y}$ defined as above.

\item[2.]
\textbf{Force $\mathsf{X}$ to be classical.} Arthur measures register $\mathsf{X}$ in the computational basis and reads the measurement outcome. This forces $\mathsf{X}$ to essentially be a classical register of bits, and destroys any entanglement or correlations between $\mathsf{X}$ and $\mathsf{Y}$.

\item[3.]
\textbf{$\mathsf{X}$ should contain probability distributions.} Arthur checks whether the content of registers $\mathsf{X}_j$ form a probability distribution $p_j$, i.e., that $\sum_{i=1}^r p_j(i)=1$. Arthur rejects if this is not the case.

\item[4.]
\textbf{Consistency check: Can the quantum states in $\mathsf{Y}$ reproduce the distributions in $\mathsf{X}$?} Arthur picks independently and uniformly at random, an index $j \in [m]$ and another index $i \in [r]$. He applies the measurement $\{ \Pi_j(i) \}_{i=1}^r$ separately to each register $\mathsf{Y}_j^1, \dots, \mathsf{Y}_j^k$, and counts the number of times outcome $i$ appears, which we denote henceforth as $n_j(i)$. Arthur rejects if
\[
\left\arrowvert \frac{n_j(i)}{k} - p_j(i) \right\arrowvert \ge \frac{1}{p}
\]
for $p$ a carefully chosen polynomial.

\item[5.]
\textbf{Run Stage 2 of the \class{BellQMA} verification and repeat for error amplification.} For each prover $j$, Arthur samples an outcome from $[r]$ according to the distribution in $(\mathsf{X}_j^1, \dots, \mathsf{X}_j^r)$, and runs Stage 2 of the \class{BellQMA} verification on the resulting set of samples. He repeats this process independently a polynomial number of times $q$, and accepts if and only if the \class{BellQMA} procedure accepts on the majority of the runs.
\end{mylist}

Let us discuss the intuition behind the verification procedure above. The key step above is Step 4, where Arthur cross-checks that the classical distributions sent in $\mathsf{X}$ really can be obtained by measuring $m$ quantum proofs, which for an honest Merlin should be unentangled. In this sense, our protocol can alternatively be viewed as using \emph{quantum} proofs ($\mathsf{Y}$) to check validity of a \emph{classical} proof ($\mathsf{X}$). Intuitively, the reason why entanglement in $\mathsf{Y}$ does not help a dishonest Merlin in Step 3 is due to the local nature of Arthur's checks/measurements. Finally, once Arthur is satisfied that $\mathsf{X}$ contains valid distributions, he runs Step 5. We remark that repetition is used here in order to boost the probability of acceptance in the $x\in A_{\yes}$ case to exponentially close to $1$, which is required to separate it from the $x \in A_{\no}$ case, where the probability of catching a dishonest Merlin is only inverse polynomially bounded away from $1$. Once such a gap exists, standard amplification techniques~\cite{KW00,MW05} can be used to further improve completeness and soundness parameters.

To formally analyze completeness and soundness of the protocol, we assign the following values to the parameters mentioned above, all of which are polynomial in $n$ in our setting:
\[
q = 50n \qquad \text{and} \qquad p = 20mr \qquad \text{and} \qquad k = 5p^3\qquad \text{and} \qquad \alpha = 20nmr.
\]

\paragraph{Completeness.} Intuitively, when $x\in A_{\yes}$, Merlin passes Step 4 with probability exponentially close to $1$ since he has no incentive to cheat --- he can send an unentangled proof in Step 1 to Arthur corresponding to the optimal proofs $\rho_j$ in the \class{BellQMA} protocol, such that the expected value of $n_j(i)/k$ is indeed $p_j(i)$. Arthur's checks in Step 4 are then independent local trials, allowing a Chernoff bound to be applied. We then show that Merlin passes each run in Step 5 with constant probability, and applying the Chernoff bound a second time yields the desired completeness exponentially close to $1$ for the protocol.

To state this formally, suppose Merlin is honest and sends registers $(\mathsf{X}, \mathsf{Y})$ in the desired form, i.e., $\mathsf{X}_j^i$ contains $p_j(i)=\ip{\Pi_j(i)}{\rho_j}$ up to $\alpha$ bits of precision, and $\mathsf{Y}_j^l$ contains $\rho_j$. Then, the expected value of the random variable $n_j(i)$ is $\E[n_j(i)] = k \ip{\Pi_j(i)}{\rho_j}$, which is equal to $k\cdot p_j(i)$ up to the error incurred by representing $p_j(i)$ using $\alpha$ bits of precision. In other words,
\begin{equation}\label{eq:precision}
    \abs{ \frac{\E[n_j(i)]}{k} - p_j(i)} < \frac{1}{2^\alpha} < \frac{1}{2p}.
\end{equation}
We can hence upper bound the probability of rejecting in Step 3 by
\[
\Pr \left[ \left\arrowvert \frac{n_j(i)}{k} - p_j(i) \right \arrowvert \ge \frac{1}{p} \right] < \Pr \left[ \left\arrowvert \frac{n_j(i)}{k} - \frac{\E[n_j(i)]}{k} \right \arrowvert \ge \frac{1}{2p} \right]  \le 2\exp\left(-\frac{5p}{4}\right)
\]
where the first inequality follows from Eq.~(\ref{eq:precision}) and the second from the Chernoff bound. Thus, Merlin passes Step 4 with probability exponentially close to $1$.

We now turn to the final step. Since $x\in A_{\yes}$, we know that the optimal distributions, denoted $q_j := \left(\ip{\Pi_j(1)}{\rho_j}, \dots, \ip{\Pi_j(r)}{\rho_j}\right)$ for $j\in[m]$, obtained in Stage 1 of the original \class{BellQMA} protocol are now accepted in Stage 2 with probability at least $2/3$. However, in our case, Merlin was only able to specify each $q_j$ up to $\alpha$ bits of precision per entry as the distributions $p_j$. To analyze how this affects the probability of acceptance, let $P_j$ and $Q_j$ be diagonal operators with entries $P_j(i,i)=p_j(i)$ and $Q_j(i,i)=\ip{\Pi_j(i)}{\rho_j}$, respectively. Letting $\Lambda_{\rm accept}$ denote the POVM element corresponding to outcome {\it accept} in Stage 2 of the BellQMA protocol, we thus bound the change in acceptance probability by:
\begin{eqnarray*}
    \left| \tr\left[\Lambda_{\rm accept} \left(\bigotimes_{j=1}^m P_j - \bigotimes_{j=1}^m Q_j\right)\right] \right| &\leq& \bigg\Vert \bigotimes_{j=1}^m P_j - \bigotimes_{j=1}^m Q_j \bigg\Vert_{\textup{tr}}\\
    &\le& \sum_{j=1}^m \trnorm{P_j - Q_j} \\
    &=& \sum_{j=1}^m \sum_{i=1}^r |p_j(i) - \ip{\Pi_j(i)}{\rho_j}|\\
    &\le& \frac{mr}{2^{20nmr}}
\end{eqnarray*}
where the first inequality follows from the fact that $|\tr(AB)| \le \snorm{A} \cdot \trnorm{B}$ and the second inequality follows from Lemma~\ref{lem:tracenorm}. Therefore, the probability of success for each of the $q$ runs of the \class{BellQMA} protocol in Step 5 is at least
\[
\left(\frac{2}{3} - \frac{mr}{2^{20nmr}}\right) > 0.6.
\]
Since each run is independent, applying the Chernoff bound yields that Arthur accepts Merlin's proof in Step 5 with probability at least $1 - 2\exp( -0.02q)$, as desired. There may be some error incurred in sampling, which can be assumed to be exponentially small so that the success probability of each run is still at least $0.6$.

\paragraph{Soundness.} We now prove that when $x\in A_{\no}$, a dishonest Merlin can win with probability at most inverse polynomially bounded away from $1$. To show this, we bound the probability of passing Step 4 by relating the quantity $p_j(i)$ to the expected value of $n_j(i)/k$, and then apply the Markov bound. The desired relationship follows by observing first that the expected value of $n_j(i)/k$ is precisely the probability of obtaining outcome $i$ when measuring proof $j$ of some (honest) unentangled strategy, followed by arguing that the distribution $p_j$ must hence be far from this latter (honest) distribution if Merlin is to pass Step 5 with probability at least $1/2$ (since $x\in A_{\no}$). Combining these facts, we find that Arthur detects a cheating Merlin with inverse polynomial probability in Step 4.

More formally, let the quantum register $\mathsf{Y}_j$ contain an arbitrary quantum state $\sigma_j$ whose reduced states in registers $\mathsf{Y}_j^l$ for $l\in[k]$ are given by $\sigma_j(l)$, and define
\[
\xi_j := \frac{1}{k}\sum_{l=1}^k\sigma_j(l).
\]
By the linearity of expectation, the expected value of the random variable $n_j(i)/k$ is
\[
\E\left[\frac{n_j(i)}{k}\right] = \frac{1}{k}\sum_{l=1}^k \ip{\Pi_j(i)}{\sigma_j(l)} = \ip{\Pi_j(i)}{\xi_j}.
\]
Our goal is to lower bound the expression
\begin{equation}
    \Pr \left[ \left\arrowvert \frac{n_j(i)}{k} - p_j(i) \right \arrowvert \ge \frac{1}{p} \right]. \label{eq:prob}
\end{equation}
To achieve this, we first substitute $p_j(i)$ above with a quantity involving $\E[n_j(i)/k]$, and then apply the Markov bound.

To relate $\E[n_j(i)/k]$ to $p_j(i)$, we first remark that in order for Merlin to pass each run of Step 5 with probability exponentially close to $1$, he must send probability distributions $p_j$, which are accepted by Stage 2 of the \class{BellQMA} verification with probability at least $1/2$. Let
\[
q_j(i):= \ip{\Pi_j(i)}{\xi_j}.
\]
Let us imagine a $\class{BellQMA}$ protocol where the $j$-th Merlin sends $\xi_j$ as his quantum proof. Since $x \in A_{\no}$, by the soundness property of the $\class{BellQMA}(m)$ proof system, the success probability of the Merlins is at most $1/3$. In other words, sampling outcomes from the probability distributions $(q_j(1), \ldots, q_j(r))$ and then running the second stage of the \class{BellQMA} verification will yield outcome {\it accept} with probability at most $1/3$. Also, observe that
\[
\E\left[\frac{n_j(i)}{k}\right] = q_j(i).
\]
It follows that by letting $P_j$ and $Q_j$ be diagonal operators with the probability vectors $p_j$ and $q_j$ on their diagonals, respectively, and $\Lambda_{\rm accept}$ the POVM element corresponding to outcome {\it accept} in Stage 2 of the BellQMA protocol, we have
\[
\frac{1}{10} < \abs{\tr\left[\Lambda_{\rm accept} \left(\bigotimes_{j=1}^m P_j - \bigotimes_{j=1}^m Q_j\right)\right]}\le \bigg\Vert \bigotimes_{j=1}^m P_j - \bigotimes_{j=1}^m Q_j \bigg\Vert_{\textup{tr}} \le \sum_{j=1}^m \trnorm{P_j - Q_j}.
\]
Here, the (loose) lower bound of $1/10$ comes from the following two observations. First, the distributions represented by the diagonal operators $Q_j$'s are derived from a \class{BellQMA} protocol and therefore achieve a success probability at most $1/3$ by the soundness property of the \class{BellQMA} verification. Second, the distributions represented by the diagonal operators $P_j$'s have to achieve a success probability strictly greater than $1/2$ per run to guarantee that Merlin wins Step 5 with probability exponentially close to $1$. Combining these two, we get that the difference between the success probabilities obtained by distributions described by operators $\set{P_j : j \in [m]}$ and $\set{Q_j: j \in [m]}$ should be at least $1/6$ modulo the error incurred due to finite precision when encoding the distributions $p_j$. The use of the constant $1/10$ overcompensates for this precision error. Hence, there exists a $j$ such that
\[
\trnorm{P_j - Q_j} = \sum_{i=1}^r |p_j(i) - q_j(i)| \ge \frac{1}{10m}
\]
implying the existence of an $i$ such that
\begin{equation}\label{eq:relation}
    |p_j(i) - q_j(i)| \ge \frac{1}{10mr}.
\end{equation}
This is our desired relationship between $p_j(i)$ and $\E[n_j(i)/k]=q_j(i)$. Note that the probability of picking pair $(i,j)$ in Step 4 is $1/mr$.

We now substitute this relationship into Eq.~(\ref{eq:prob}) and apply the Markov bound. Specifically, choose $i$ and $j$ as in Eq.~(\ref{eq:relation}), and assume that $p_j(i) > \ip{\Pi_j(i)}{\xi_j}$. Then, we have
\[
\Pr \left[ \left\arrowvert \frac{n_j(i)}{k} - p_j(i) \right\arrowvert < \frac{1}{p} \right] < \Pr \left[ \frac{n_j(i)}{k} - \E\left[\frac{n_j(i)}{k}\right] > \frac{1}{10mr} - \frac{1}{p} \right] \le 1-\frac{1}{2p}.
\]

\noindent The case of $p_j(i) < \ip{\Pi_j(1)}{\xi_j}$ is similar. We conclude that a dishonest Merlin is caught in Step 4 with probability at least $1/2p$. Therefore, the probability that Arthur proceeds to Step 5 is upper bounded by
\[
\left(\frac{1}{mr}\right)\left(1 - \frac{1}{40mr}\right) + \left(1 - \frac{1}{mr}\right)(1) = 1 - \frac{1}{40m^2r^2}
\]
where the first term represents the case where Arthur selects the correct pair $(i,j)$ to check, and the second term the complementary case, in which we assume the cheating prover can win with probability $1$. Hence the overall success probability of a dishonest Merlin is at most $1 - 1/40m^2r^2$, which is bounded away from $1$ by an inverse polynomial.

Finally, as mentioned before, since $m$ and $r$ are polynomially bounded functions, we have that the completeness is exponentially close to $1$, while the soundness is bounded away from $1$ by an inverse polynomial. By known amplification techniques for \class{QMA} protocols~\cite{KW00,MW05}, one can amplify the completeness and soundness errors to be exponentially close to 0. This proves our desired containment.

\section{Perfect parallel repetition for $\class{SepQMA}(\poly)$}\label{scn:parrep}

We now show Theorem~\ref{thm:parrep}, i.e., that the class $\class{SepQMA}(m)$ admits perfect parallel repetition. Before we proceed, recall that the closed convex cone $\sep{\X_1, \dots, \X_m}$ is defined to contain operators of the form
\[
\sum_{i=1}^k P_{1}(i) \otimes \dots \otimes P_{m}(i)
\]
where $P_j(i) \in \pos{\X_j}$, for every $j \in [m]$ and $i \in [k]$.
This is the cone of interest and it is known to be closed and convex with non-empty interior. Given $C$ to be the measurement operator corresponding to outcome {\it accept}, the maximum success probability of the Merlins in any $\class{QMA}(m)$ protocol can be written as the maximum of $\ip{\rho}{C}$, where $\rho$ is a density operator in $\sep{\X_1, \dots, \X_m}$. By standard convexity argument, one can always assume that the maximum is achieved by a pure product state.

For the remainder of the section, it will be convenient for us to distinguish  two instances of $\class{SepQMA}(m)$ protocols as the \emph{first} and \emph{second} protocol. For the first $\class{SepQMA}(m)$ protocol we can write the maximum acceptance probability as the optimal value of the primal problem in the following primal-dual pair (where the operator  $C_1$ is Arthur's POVM element corresponding to outcome {\it accept}):
\begin{center}
  \begin{minipage}{2.5in}
    \centerline{\underline{Primal problem ($\textup{P}_1$)}}\vspace{-7mm}
    \begin{align*}
			\text{maximize:}\quad & \ip{\rho_1}{C_1} \\
  		\text{subject to:}\quad & \tr(\rho_1) = 1,\\
  		& \rho_1 \in \sep{\X_1, \dots, \X_m},
  	\end{align*}
  \end{minipage}
  \hspace*{25mm}
  \begin{minipage}{2.5in}
    \centerline{\underline{Dual problem ($\textup{D}_1$)}}\vspace{-7mm}
		\begin{align*}
			\text{minimize:}\quad & t_1\\
  		\text{subject to:}\quad & t_1\I_{\X} = C_1 + W_1,\\
  		& W_1 \in \sep{\X_1, \ldots, \X_m}^{\ast},
  	\end{align*}	
  \end{minipage}
\end{center}
where $\X$ denotes $\X_1 \otimes \cdots \otimes \X_m$. The use of ``maximum'' and ``minimum'' is justified in the above programs since
\[
\overline{\rho}_1 = \frac{\I_{\X}}{\dim(\X)} \qquad \text{and} \qquad (\overline{t}_1, \overline{W}_1) = (2, 2\I_{\X} - C_1)
\]
are strictly feasible solutions for $(\textup{P}_1)$ and $(\textup{D}_1)$, respectively~\cite{GB02, GB03, GB05}. Hence, by Theorem~\ref{lem:stduality2}, strong duality holds for both problems, i.e., both problems attain an optimal solution and the optimal values are the same. We note that the the dual cone contains the set of \emph{entanglement witnesses} in the theory of entanglement, see~\cite{HHHH09}. We can similarly formulate the acceptance probability of the second protocol as
\begin{center}
  \begin{minipage}{2.5in}
    \centerline{\underline{Primal problem ($\textup{P}_2$)}}\vspace{-7mm}
    \begin{align*}
			\text{maximize:}\quad & \ip{\rho_2}{C_2} \\
  		\text{subject to:}\quad & \tr(\rho_2) = 1,\\
  		& \rho_2 \in \sep{\Y_1, \dots, \Y_m},
  	\end{align*}
  \end{minipage}
  \hspace*{25mm}
  \begin{minipage}{2.5in}
    \centerline{\underline{Dual problem ($\textup{D}_2$)}}\vspace{-7mm}
		\begin{align*}
			\text{minimize:}\quad & t_2\\
  		\text{subject to:}\quad & t_2\I_{\Y} = C_2 + W_2,\\
  		& W_2 \in \sep{\Y_1, \ldots, \Y_m}^{\ast},
  	\end{align*}	
  \end{minipage}
\end{center}
where $\Y$ denotes $\Y_1 \otimes \cdots \otimes \Y_m$. Since we are considering \class{SepQMA} protocols it holds that
\[
C_1 \in \sep{\X_1, \dots, \X_m} \qquad \text{and} \qquad C_2 \in \sep{\Y_1, \dots, \Y_m}.
\]

Given the two cone programs above, the maximum acceptance probability of the two-fold repetition of the protocol can hence be expressed as

\text{}

\centerline{\underline{Primal problem ($\textup{P}$)}\hspace{57mm}\underline{Dual problem ($\textup{D}$)}\hspace{3mm}}\vspace{-7mm}
\begin{align*}
	\text{maximize:}\quad & \ip{\rho}{C_1 \otimes C_2} \hspace{38mm}&\text{minimize:}\quad & t\\
  \text{subject to:}\quad & \tr(\rho) = 1, & \text{subject to:}\quad & t \, \I_{\X \otimes \Y} = C_1 \otimes C_2 + W,\\
  & \rho \in \sep{\X_1 \otimes \Y_1, \dots, \X_m \otimes \Y_m}, & & W \in \sep{\X_1 \otimes \Y_1, \ldots, \X_m \otimes \Y_m}^{\ast}.
\end{align*}

\noindent Note that the operators $\rho$ and $W$ are elements of $\herm{\X_1 \otimes \cdots \otimes \X_m \otimes \Y_1 \otimes \cdots \otimes \Y_m}$.

To show Theorem~\ref{thm:parrep}, observe that if $\rho_1$ and $\rho_2$ are any respective optimal solutions of ($\textup{P}_1$) and ($\textup{P}_2$), then $\rho_1 \otimes \rho_2$ is a feasible solution of ($\textup{P}$). Therefore the optimal value of ($\textup{P}$) is at least the product of the optimal values of ($\textup{P}_1$) and ($\textup{P}_2$). It remains to show that in fact \emph{no} other strategy for the prover can perform better than this honest strategy. To do so, we demonstrate a dual feasible solution for ($\textup{D}$) attaining this same objective value.

More formally, let $(t_1, W_1)$ and $(t_2, W_2)$ be respective dual optimal solutions of $(\textup{D}_1)$ and $(\textup{D}_2)$. By strong duality, $t_1$ is the optimal value of $(\textup{P}_1)$ and $t_2$ is the optimal value of $(\textup{P}_2)$. We show that $t_1 \cdot t_2$ is an upper bound on the optimal value of $(\textup{P})$ by exhibiting a solution $(t_1 \cdot t_2, W)$ which is feasible in $(\textup{D})$, for some $W \in \sep{\X_1 \otimes \Y_1, \ldots, \X_m \otimes \Y_m}^{\ast}$. We first prove the following useful lemma.

\begin{lemma}\label{lem:dualconstruction}
For complex Euclidean spaces $\X_1, \ldots, \X_m$ and $\Y_1, \ldots, \Y_m$, the following two containments hold:
\begin{itemize}
\item $\sep{\X_1, \ldots, \X_m}^{\ast} \otimes \sep{\Y_1, \ldots, \Y_m} \subseteq \sep{\X_1 \otimes \Y_1, \ldots, \X_m \otimes \Y_m}^{\ast}$, and
\item $\sep{\X_1, \ldots, \X_m} \otimes \sep{\Y_1, \ldots, \Y_m}^{\ast} \subseteq \sep{\X_1 \otimes \Y_1, \ldots, \X_m \otimes \Y_m}^{\ast}$.
\end{itemize}
\end{lemma}

\begin{proof}
We prove the first condition as the second is nearly identical. Fix $W \in \sep{\X_1, \ldots, \X_m}^{\ast}$ and $C \in \Sep(\Y_1, \ldots, \Y_m)$. Then for $S \in \Sep(\X_1 \otimes \Y_1, \ldots, \X_m \otimes \Y_m)$, we have
\[
\ip{W \otimes C}{S} = \inner{W}{\tr_{\Y}\left[S(\I_{\X} \otimes C)\right]} \geq 0
\]
if $\tr_{\Y}\left[S(\I_{\X} \otimes C)\right] \in \Sep(\X_1, \ldots, \X_m)$. Therefore, it suffices to prove that 
\[ \tr_{\Y}\left[S(\I_{\X} \otimes C)\right] \in \Sep(\X_1, \ldots, \X_m). \]  To this end, let
\[
S = \sum_{i=1}^k \bigotimes_{l=1}^m \rho_i(l) \qquad \text{ and } \qquad C = \sum_{j=1}^{k'} \bigotimes_{l=1}^m \sigma_j(l)
\]
where $\rho_i(l) \in \pos{\X_l \otimes \Y_l}$ and $\sigma_j(l) \in \pos{\Y_l}$ for all $i \in [k]$, $j \in [k']$, and $l \in [m]$. Now we can write $\tr_{\Y}\left[ S\left(\I_{\X} \otimes C \right) \right]$ as
\[ \tr_{\Y} \left[ \left( \sum_{i=1}^k \bigotimes_{l=1}^m \rho_i(l)\right) \left( \I_{\X} \otimes \sum_{j=1}^{k'} \bigotimes_{l=1}^m \sigma_j(l) \right) \right] = \sum_{i=1}^k \sum_{j=1}^{k'} \bigotimes_{k=1}^m \tr_{\Y_k} \left[ \rho_i(k) \left(\I_{\X_k} \otimes \sigma_j(k) \right) \right].
\]
Hence, $\tr_{\Y}\left[ S\left(\I_{\X} \otimes C \right) \right]\in\sep{\X_1, \ldots, \X_m}$  since $\tr_{\Y_k} \left[ \rho_i(k) \left(\I_{\X_k} \otimes \sigma_j(k) \right) \right]$ is positive semidefinite for all $i,j,k$. The latter follows since for positive semidefinite $A_{\X\otimes\Y}$ and $B_{\Y}$,
\[
    \tr_{\Y}[A_{\X\otimes\Y}(I_{\X}\otimes B_{\Y})] = \tr_{\Y}[(I_{\X}\otimes B^{\frac{1}{2}}_{\Y})A_{\X\otimes\Y}(I_{\X}\otimes B^{\frac{1}{2}}_{\Y})]\succeq 0,
\]
because the partial trace preserves positive semidefiniteness. This concludes the proof. 
\end{proof}

We now use Lemma~\ref{lem:dualconstruction} to construct two operators in $\sep{\X_1 \otimes \Y_1, \ldots, \X_m \otimes \Y_m}^{\ast}$, the appropriate convex combination of which is the dual feasible solution we are seeking. Specifically, observe first that since for the two instances of the $\class{SepQMA}(m)$ protocol, we have 
\[ C_1 \in \sep{\X_1, \ldots, \X_m} \qquad \text{and} \qquad C_2 \in \sep{\Y_1, \ldots, \Y_m}, \] 
and since $\I_{\X}$ and $\I_{\Y}$ are fully separable operators, it follows that
\[
t_1 \I_{\X} + C_1 \in \Sep(\X_1, \ldots, \X_m) \qquad \text{and} \qquad t_2 \I_{\Y} + C_2 \in \Sep(\Y_1, \ldots, \Y_m)
\]
for all $t_1, t_2 \geq 0$. Using Lemma~\ref{lem:dualconstruction}, we thus obtain operators
\begin{equation}\label{eq:dualeq1}
(t_1 \I_{\X} - C_1) \otimes (t_2 \I_{\Y} + C_2) \in \sep{\X_1 \otimes \Y_1, \ldots, \X_m \otimes \Y_m}^{\ast}
\end{equation}
and
\begin{equation} \label{eq:dualeq2}
(t_1 \I_{\X} + C_1) \otimes (t_2 \I_{\Y} - C_2) \in \sep{\X_1 \otimes \Y_1, \ldots, \X_m \otimes \Y_m}^{\ast}
\end{equation}
where $t_1 \I_{\X} - C_1\in \sep{\X_1, \ldots, \X_m}^{\ast}$ by the constraints of $(\textup{D}_1)$, and similarly for $t_2 \I_{\Y} - C_2$. Since $\sep{\X_1 \otimes \Y_1, \ldots, \X_m \otimes \Y_m}^{\ast}$ is a convex cone, it follows that the average of Eqs.~\eqref{eq:dualeq1} and~\eqref{eq:dualeq2} yields the desired operator
\[
W := t_1 \cdot t_2 \, \I_{\X \otimes \Y} -  C_1 \otimes C_2 \in \sep{\X_1 \otimes \Y_1, \ldots, \X_m \otimes \Y_m}^{\ast}.
\]
We conclude that $\left(t_1 \cdot t_2, W \right)$ is a feasible solution of the dual problem ($\textup{D}$) with objective value $t_1 \cdot t_2$ as desired. This concludes the proof of Theorem~\ref{thm:parrep}.

We note that there are instances of $\class{QMA}(\poly)$ protocols, which are not $\class{SepQMA}(\poly)$ protocols, that admit perfect parallel repetition. Although this fact is known in the literature (see Harrow and Montanaro~\cite{HM10} for details), we provide a concrete example below.

First, note that the maximum acceptance probability of Arthur in a $\class{QMA}(m)$ protocol is upper bounded by $\snorm{C}$, where $C$ is the accepting measurement operator. Now, consider the two-qubit POVM operator
\[
C := \frac{1}{2} \ket{00}\bra{00} +  \frac{1}{2}\ket{\Psi^+}\bra{\Psi^+}
\]
where
\[
\ket{\Psi^{+}} := \frac{1}{\sqrt{2}} \ket{01} + \frac{1}{\sqrt{2}} \ket{10}.
\]
We can easily check that $C$ has two eigenvalues, $0$ and $1/2$, and two principal eigenvectors $\ket{00}$ and $\ket{\Psi^+}$, one of which is a product state. It follows that the maximum acceptance probability is $1/2$. By the multiplicative property of the infinity-norm under tensor products, it holds that the maximum acceptance probability of the $k$-fold repetition is exactly $1/2^k$.

We now argue that $C$ is not a separable operator. Suppose for the sake of contradiction that $C$ can be written as
\[
\sum_{i=1}^n \rho_i \otimes \sigma_i
\]
for some $\rho_i, \sigma_i \in \Pos(\C^2)$. Then we have
\[
0 = \inner{C}{\ket{11}\bra{11}} = \sum_{i=1}^n \bra{1} \rho_i \ket{1} \bra{1} \sigma_i \ket{1}
\]
which implies $\rho_i \ket{1} = 0$ or $\sigma_i \ket{1} = 0$ for all $i \in [n]$. This leads to the contradiction
\[
\frac{1}{4} = \inner{C}{\ket{01} \bra{10}} = \sum_{i=1}^n \bra{1} \rho_i \ket{0} \bra{0} \sigma_i \ket{1} = 0.
\]
Alternatively, one can show that $C$ is not separable by observing that $C$ has a non-positive partial transpose~\cite{P96,HHH96}.

\section{Conclusions and open problems} \label{scn:conclusions}

In this paper, we have studied three variants of multi-prover quantum Merlin-Arthur proof systems. We first showed that a system with polynomially many provers is indeed strictly more powerful than a single prover system if messages are restricted to be logarithmic in length, unless $\class{BQP}=\class{MQA}$. We next showed that polynomially many provers do not provide additional expressive power over a single prover in the setting where the verifier is restricted to first applying unentangled and non-adaptive measurements with at most a polynomial number of outcomes per proof. Both of these questions make steps towards understanding the major open question of whether $\class{QMA}$ with polynomially many provers is more powerful than $\class{QMA}$. Finally, we used cone programming duality to give an alternate proof of the fact that perfect parallel repetition holds whenever a $\class{QMA}$ verifier's POVM element corresponding to {\it accept} is a fully separable operator.

A consequence of our first result is that the two variants of the class $\class{QMA}(\poly)$, where Merlins send logarithmic-size proofs and Merlins send constant-size proofs are equal. A natural question concerning our first result is to understand the expressive power of the variant of $\class{QMA}(\poly)$, where Merlins are restricted to send $\poly\log(\abs{x})$ qubits to Arthur. Another open question concerning the results presented in this paper is the relationship between $\class{BellQMA}(\poly)$ and \class{QMA}. We believe that understanding the complexity of $\class{BellQMA}$ protocols, or more generally \class{LOCCQMA} protocols, will shed new light on the bigger question pertaining to \class{QMA}(2) and \class{QMA}. Another avenue of interest is to find further applications of the cone programming characterization of multi-prover quantum Merlin-Arthur proof systems. A straightforward question concerning the parallel repetition result presented in this paper is to investigate whether cone programming duality can be used to analyze the product state test in the Ref.~\cite{HM10}. Another question one can ask is to find other classes of $\class{QMA}(m)$ protocols that admit a perfect parallel repetition theorem.

\section*{Acknowledgements}

We thank Richard Cleve, Tsuyoshi Ito, Iordanis Kerenidis, Ashwin Nayak, Oded Regev, and Levent Tun{\c c}el for insightful discussions. Most of this work was completed at IQC at the University of Waterloo during the authors' graduate studies. We also thank the EU-Canada Exchange Program and LIAFA, Paris for their hospitality, where part of this work was completed. SG acknowledges support from NSERC, NSERC MSFSS, the David R.~Cheriton Graduate Scholarship program, the President's Graduate Scholarship, and CIFAR. JS acknowledges support from NSERC, MITACS, ERA (Ontario), and the President's Graduate Scholarship. SU acknowledges support in parts from CIFAR, MITACS, NSERC, Ontario's Ministry of Research and Innovation, QuantumWorks, the U.S. A.R.O., David R.~Cheriton Graduate Scholarship, the Mike and Ophelia Graduate Fellowship, and internal grants of Centre for Quantum Technologies.

\bibliographystyle{alpha}
\bibliography{BibQMA}

\end{document}